%% file: wads-2015.tex
\documentclass{llncs}

\usepackage{capt-of,comment}
\usepackage{makeidx}  
\usepackage[ruled, vlined, linesnumbered, commentsnumbered]{algorithm2e}
\usepackage{mathdots}
\usepackage{amsmath}
\usepackage{amssymb}
\usepackage{mathtools}
\usepackage{etoolbox}
\usepackage{cite}
\usepackage{lipsum}
\usepackage[pdftex]{graphicx}
\usepackage{paralist}
\usepackage[labelformat=simple,font=footnotesize]{subcaption}
\usepackage{tikz}
\usepackage{pgfplots}
\usepackage{needspace}

\usepackage{scalerel,stackengine}
\stackMath
\newcommand\reallywidehat[1]{%
\savestack{\tmpbox}{\stretchto{%
  \scaleto{%
    \scalerel*[\widthof{\ensuremath{#1}}]{\kern-.6pt\bigwedge\kern-.6pt}%
    {\rule[-\textheight/2]{1ex}{\textheight}}
  }{\textheight}%
}{0.5ex}}%
\stackon[1pt]{#1}{\tmpbox}%
}
\usetikzlibrary{shapes,shapes.geometric,arrows,fit,matrix,positioning,decorations.pathreplacing,calc,fpu,decorations.pathmorphing,decorations.markings, decorations}
\newcommand{\natnum}{\mathbb{N}}

\newcommand{\ints}{\mathbb{Z}}

\newcommand{\repfact}{\rho}
\newcommand{\lleq}{\leq_L}
\newcommand{\lgeq}{\geq_L}

\DeclarePairedDelimiter{\ceil}{\lceil}{\rceil}
\DeclarePairedDelimiter{\floor}{\lfloor}{\rfloor}
\newcommand{\ceilfrac}[2]{\ceil{\tfrac{#1}{#2}}}
\newcommand{\floorfrac}[2]{\floor{\tfrac{#1}{#2}}}

\SetKwFunction{transf}{Transform}
\SetKwFunction{mkPseudo}{Make-Pseudonode}
\SetKwFunction{filled}{Filled}
\SetKwFunction{getFilled}{Get-Filled}

\AtBeginEnvironment{theorem}{\Needspace{5\baselineskip}}

\begin{document}

\mainmatter              
\title{Algorithms for Replica Placement in High-Availability Storage }
\titlerunning{Survivable Replica Placement}  
%
\author{K. Alex Mills, R. Chandrasekaran, Neeraj Mittal}
\authorrunning{K. Alex Mills et al.} 
%
\tocauthor{K. Alex Mills, R. Chandrasekaran, Neeraj Mittal}
\institute{Department of Computer Science \\The University of Texas at Dallas, Richardson Texas, USA\\
\email{\{k.alex.mills,chandra,neerajm\}@utdallas.edu}}

\maketitle              

\begin{abstract}
A new model of causal failure is presented and used to solve a novel replica placement problem in data centers. The model describes dependencies among system components as a directed graph. A replica placement is defined as a subset of vertices in such a graph. A criterion for optimizing replica placements is formalized and explained. In this work, the optimization goal is to avoid choosing placements in which a single failure event is likely to wipe out multiple replicas. Using this criterion, a fast algorithm is given for the scenario in which the dependency model is a tree. The main contribution of the paper is an $O(n + \repfact^2)$ dynamic programming algorithm for placing $\repfact$ replicas on a tree with $n$ vertices. This algorithm exhibits the interesting property that only \emph{two} subproblems need to be recursively considered at each stage. An $O(n^2 \repfact)$ greedy algorithm is also briefly reported.
\end{abstract}

\section{Introduction}

With the surge towards the cloud, our websites, services and data are increasingly being hosted by third-party data centers. These data centers are often contractually obligated to ensure that data is rarely, if ever unavailable. One cause of unavailability is co-occurring component failures, which can result in outages that can affect millions of websites \cite{Ver:2013:Blog}, and can cost millions of dollars in profits \cite{Ple:2013:Blog}. An extensive one-year study of availability in Google's cloud storage infrastructure showed that such failures are relatively harmful. Their study emphasizes that ``correlation among node failure dwarfs all other contributions to unavailability in our production environment" \cite{ForFra+:2010:OSDI}.

We believe that the correlation found among failure events arises due to dependencies among system components. Much effort has been made in the literature to produce quality statistical models of this correlation. But in using such models researchers do not make use of the fact that these dependencies can be explicitly modeled, since they are known to the system designers. In contrast, we propose a model wherein such dependencies are included, and demonstrate how an algorithm may make use of this information to optimize placement of data replicas within the data center.

To achieve high availability, data centers typically store multiple replicas of data to tolerate the potential failure of system components. This gives rise to a \emph{placement problem}, which, broadly speaking, involves determining which subset of nodes in the system should store a copy of a given file so as to maximize a given objective function (\emph{e.g.}, reliability, communication cost, response time, or access time).  While our focus is on replica placements, we note that our model could also be used to place replicas of other system entities which require high-availability, such as virtual machines and mission-critical tasks.

\begin{figure}[b]
\begin{subfigure}[b]{0.5\textwidth}
\input{informal}
\vspace{-0.5cm}
\caption{Scenario I} \label{fig:scenarioI}
\end{subfigure}
\begin{subfigure}[b]{0.5\textwidth}
\input{informalb}
\vspace{-0.5cm}
\caption{Scenario II} \label{fig:scenarioII}
\end{subfigure}
\caption{Two scenarios represented by directed trees.}
\label{fig:scenarios}
\end{figure}
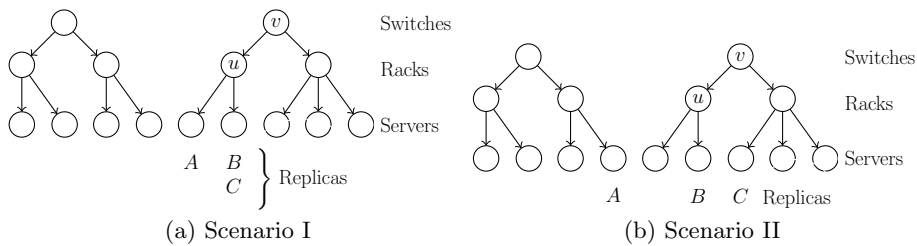


In this work, we present a new model for causal dependencies among failures, and a novel algorithm for optimal replica placement in our model. An example model is given as Fig. \ref{fig:scenarios}, in which three identical replicas of the same block of data are distributed on servers in a data center. Each server receives power from a surge protector which is located on each server rack. In Scenario I, each replica is located on nodes which share the same rack. In Scenario II, each replica is located on separate racks. As can be seen from the diagram of Scenario I (Fig. \ref{fig:scenarioI}), a failure in the power supply unit (PSU) on a single rack could result in a situation where every replica of a data block is completely unavailable, whereas in Scenario II, (Fig. \ref{fig:scenarioII}) three PSUs would need to fail in order to achieve the same result. In practice, Scenario I is avoided by ensuring that each replica is placed on nodes which lie on separate racks. This heuristic is already part of known best-practices. Our observation is that this simple heuristic can be suboptimal under certain conditions. For example, consider a failure in the aggregation switch which services multiple racks. Such a failure could impact the availability of every data replica stored on the rack. Moreover, this toy example only represents a small fraction of the number of events that could be modeled in a large data center.


While many approaches for replica placement have been proposed, our approach of modeling causal dependencies among failure events appears to be new. Other work on reliability in storage area networks has focused on objectives such as mean time to data loss \cite{DBLP:conficdcsLianCZ05,lian-chen-zhang}. These exemplify an approach towards correlated failure which we term ``measure-and-conquer''. In measure-and-conquer approaches, a measured degree of correlation is given as a parameter to the model. In contrast, we model explicit causal relations among failure events which we believe  give rise to the correlation seen in practice. In \cite{DBLP:conficdcsLianCZ05} the authors consider high-availability replica placement, but are primarily focused on modeling the effects of repair time. Later work \cite{lian-chen-zhang} begins to take into account information concerning the network topology, which is a step towards our approach. Similar measure-and-conquer approaches are taken in \cite{ForFra+:2010:OSDI,BakWyl+:2002:TR, WeaMos+:2002:SRDS, NatYu+:2006:NSDI}. More recently, Pezoa and Hayat \cite{Pezoa} have presented a model in which spatially correlated failures are explicitly modeled. However, they consider the problem of task allocation, whereas we are focused on replica placement. In the databases community, work on replica placement primarily focuses on finding optimal placements in storage-area networks with regard to a particular distributed access model or mutual exclusion protocol \cite{HuJia+:2001:JPDC, SheWu:2001:TCS, ZhaWu+:2009:JPDC}. In general, much of the work from this community focuses on specialized communication networks and minimizing communication costs --- system models and goals which are substantially different from our own.

Recently, there has been a surge of interest in computer science concerning cascading failure in networks \cite{BluEas+:2011:FOCS, NieLui+:2014:IPL, KimDob:2010:TransRel, ZhuYan+:2014:TPDS}. While our model is most closely related to this work, the existing literature is primarily concerned with applications involving large graphs intended to capture the structure of the world-wide web, or power grids. The essence of all these models is captured in the \textit{threshold cascade model} \cite{BluEas+:2011:FOCS}. This model consists of a directed graph in which each node $v$ is associated with a threshold, $\ell(v) \in \natnum^+$. A node $v$ experiences a cascading failure if at least $\ell(v)$ of its incoming neighbors have failed. This model generalizes our own, wherein we pessimistically assume that $\ell(v) = 1$ for all nodes $v$. Current work in this area is focused on network design \cite{BluEas+:2011:FOCS}, exploring new models  \cite{NieLui+:2014:IPL, KimDob:2010:TransRel}, and developing techniques for adversarial analysis \cite{ZhuYan+:2014:TPDS}. To our knowledge, no one has yet considered the problem of replica placement in such models.

\section{Model}\label{sec:model}

We model dependencies among failure events as a directed graph, where nodes represent failure events, and a directed edge from $u$ to $v$ indicates that the occurrence of failure event $u$ could trigger the occurrence of failure event $v$. We refer to this graph as the \emph{failure model}

Given such a graph as input, we consider the problem of selecting nodes on which to store data replicas. Roughly, we define a \emph{placement problem} as the problem of selecting a subset of these vertices, hereafter referred to as a \emph{placement}, from the failure model so as to satisfy some safety criterion. In our application, only those vertices which represent storage servers are candidates to be part of a placement. We refer to such vertices as \emph{placement candidates}. Note that the graph also contains vertices representing other types of failure events, which may correspond to real-world hardware unsuitable for storage (such as a ToR switch), or even to abstract events which have no associated physical component. In most applications, the set of placement candidates forms a subset of the set of vertices.

More formally, let $E$ denote the set of failure events, and $C$ denote the set of placement candidates. We are interested in finding a \emph{placement} of size $\repfact$, which is defined to be a set $P \subseteq C$, with $|P| = \repfact$. Throughout this paper we will use $P$ to denote a placement, and $\repfact$ to denote its size. We consistently use $C$ to denote the set of placement candidates, and $E$ to denote the set of failure events.

Let $G = (V,A)$ be a directed graph with vertices in $V$ and edges in $A$. The vertices represent both events in $E$ and candidates in $C$, so let \mbox{$V = E \cup C$}. A directed edge between events $e_1$ and $e_2$ indicates that the occurrence of failure event $e_1$ can trigger the occurrence of failure event $e_2$. A directed edge between event $e$ and candidate $c$ indicates that the occurrence of event $e$ could compromise candidate $c$. We will assume failure to act transitively. That is, if a failure event occurs, all failure events reachable from it in $G$ also occur. This a pessimistic assumption which leads to a conservative interpretation of failure.

We now define the notions of \textit{failure number} and \textit{failure aggregate}.

\begin{definition}\label{def-failure}
Let $e \in E$. The \textit{failure number} of event $e$, denoted $f(e,P)$, for a given placement $P$, is defined as the number of candidates in $P$ whose correct operation could be compromised by occurrence of event $e$. In particular, $$f(e,P) = | \{ p \in P \mid p \text{ is reachable from } e \text{ in } G  \}|.$$
\end{definition}
As an example, node $u$ in Fig. \ref{fig:scenarios} has failure number $3$ in Scenario I, and failure number $1$ in Scenario II. The following property is an easy consequence of the above definition. A formal proof can be found in the appendix.

\begin{property}\label{lem-desc}
For any placement $P$ of replicas in tree $T$, if node $i$ has descendant $j$, then $f(j, P) \leq f(i, P)$.
\end{property}

The failure number captures a conservative criterion for a safe placement. Intuitively, we consider the worst case scenario, in which every candidate which \emph{could} fail due to an occurring event \emph{does} fail. Our goal is to find a placement which does not induce large failure numbers in any event. 
To aggregate this idea across all events, we define \textit{failure aggregate}, a measure that accounts for the failure number of every event in the model.

\begin{definition}
The \emph{failure aggregate} of a placement $P$ is a vector in $\mathbb{N}^{\repfact+1}$, denoted $\vec{f}(P)$, where \mbox{$\vec{f}(P) := \langle p_\repfact, ..., p_1, p_0\rangle$}, and each $p_i := \left| \big\{ e \in E \mid f(e, P) = i\big\} \right|$.
\end{definition}
In Fig. \ref{fig:scenarios}, node $v$ has failure aggregate $\langle 2, 0, 0, 1 \rangle$ in Scenario I and failure aggregate $\langle 1, 0, 2, 0 \rangle$ in Scenario II. Failure aggregate is also computed in Fig. \ref{fig-balanced-survivalnums}.

In all of the problems considered in this paper, we are interested in optimizing $\vec{f}(P)$. When optimizing a vector quantity, we must choose a meaningful way to totally order the vectors. In the context of our problem, we find that ordering the vectors with regard to the \emph{lexicographic order} is both meaningful and convenient. The lexicographic order $\leq_L$ between $\vec{f}(P) = \langle p_\repfact, ..., p_1, p_0\rangle$ and $\vec{f}(P') = \langle p'_\repfact, ..., p'_1, p'_0\rangle$ is defined via the following formula:
$$\vec{f}(P) \leq_L \vec{f}(P') \iff \exists~ m > 0,  ~\forall ~i > m  \big[p_i = p'_i \wedge p_m \leq p'_m\big]. $$
To see why this is desirable, consider a placement $P$ which lexicominimizes $\vec{f}(P)$ among all possible placements. Such a placement is guaranteed to minimize $p_\repfact$, i.e. the number of nodes which compromise \emph{all} of the entities in our placement. Further, among all solutions minimizing $p_\repfact$, $P$ also minimizes $p_{\repfact-1}$, the number of nodes compromising \emph{all but one} of the entities in $P$, and so on for $p_{\repfact-2}, p_{\repfact-3},..., p_{0}$. Clearly, the lexicographic order nicely prioritizes minimizing the entries of the vector in an appealing manner. 

Throughout the paper, any time a vector quantity is maximized or minimized, we are referring to the maximum or minimum value in the lexicographic order. We will also use $\vec{f}(P)$ to denote the failure aggregate, and $p_i$ to refer to the $i^{th}$ component of $\vec{f}(P)$, where $P$ can be inferred from context.

In the most general case, we could consider the following problem.
\begin{problem}\label{prob:additive-function}
Given graph $G = (V,A)$ with $V = C \,\cup\, E$, and positive integer $\repfact$ with $\repfact < |C|$, find a placement $P \subseteq C$ with $|P| = \repfact$ such that $\vec{f}(P)$ is lexicominimum.
\end{problem}
Problem \ref{prob:additive-function} is NP-hard to solve, even in the case where $G$ is a bipartite graph. In particular, a reduction to independent set can be shown. However, the problem is tractable for special classes of graphs, one of which is the case wherein the graph forms a directed, rooted tree with leaf set $L$ and $C = L$. Our main contribution in this paper is a fast algorithm for solving Problem \ref{prob:additive-function} in such a case. We briefly mention a greedy algorithm which solves the problem on $O(n^2\repfact)$ time. However, since $n \gg \repfact$ in practice our result of an $O(n + \repfact^2)$ algorithm is much preferred. 

\subsection{An $O(n^2\repfact)$ Greedy Algorithm}

The greedy solution to this problem forms a partial placement $P'$, to which new replicas are added one at a time, until $\repfact$ replicas have been placed overall. $P'$ starts out empty, and at each step, the leaf $u$ which lexicominimizes $\vec{f}(P' \cup \{u\})$ is added to $P'$. This greedy algorithm correctly computes an optimal placement, however its running time is $O(n^2\repfact)$ for a tree of unbounded degree. This running time comes about since each iteration requires visiting $O(|L|)$ leaves for inclusion. For each leaf $q$ which is checked, every node on a path from $q$ to the root must have its failure number computed. Both the length of a leaf-root path and the number of leaves can be bounded by $O(n)$ in the worst case, yielding the result. 

That the greedy algorithm works correctly is not immediately obvious. It can be shown via an exchange argument that each partial placement found by the greedy algorithm is a subset of some optimal placement. This is the content of Theorem 1 below.

To establish the correctness of the greedy algorithm, we first introduce some notation. For a placement $P$ and $S \subseteq V$, let $\vec{f}(S,P) = \langle g_\repfact, g_{\repfact - 1}, ..., g_1, g_0 \rangle$ where \mbox{$g_i := | \{ x \in S \mid f(x,P) = i \} |$}. Intuitively, $\vec{f}(S,P)$ gives the failure aggregate for all nodes in set $S \subseteq V$. We first establish the truth of two technical lemmas before stating and proving Theorem \ref{thm-greedy}.

\begin{lemma}\label{lem-path-ineq}
Let $r$ be the root of a failure model given by a tree. Given $P \subseteq C$, $a,b \in C - P$. If $f(r\rightsquigarrow a, P) <_L f(r\rightsquigarrow b, P)$ then $f(P \cup \{a\}) <_L f(P \cup \{b\})$.
\end{lemma}

\begin{proof}
Suppose $f(r \rightsquigarrow a, P) <_L f(r \rightsquigarrow b, P)$. Let nodes on the paths from $r$ to $a$ and from $r$ to $b$ be labeled as follows:
$$r \rightarrow a_1 \rightarrow a_2 \rightarrow ... \rightarrow a_n \rightarrow a$$
$$r \rightarrow b_1 \rightarrow b_2 \rightarrow ... \rightarrow b_m \rightarrow b$$
We proceed in two cases. 

In the first case, there is some $1 \leq i \leq \min (m,n)$  for which $f(a_i, P) < f(b_i, P)$. Let $i$ be the minimum such index, and let $f(b_i, P) = k$. Clearly, \mbox{$f(P \cup \{a\})_k < f(P\cup \{b\})_k$}, since $P \cup \{b\}$ counts $b_i$ as having survival number $k$ and $P \cup\{a\}$ does not. Moreover, since $f(a_\ell, P) = f(b_\ell, P)$ for all $\ell < i$, we have that for all $j > k$, $f(P \cup \{a\})_j = f(P \cup \{b\})_j$ by Property \ref{prob:additive-function}.

In the second case, $f(a_i, P) \geq f(b_i, P)$ for all $1 \leq i \leq \min(m,n)$. 
In this case, if $f(a_i, P) > f(b_i, P)$ for some $i$, the only way we could have $f(r \rightsquigarrow a, P) <_L f(r\rightsquigarrow b, P)$ is if there is some $j > i$ with $f(a_j, P) < f(b_j, P)$, but this is a contradiction. Therefore, $f(a_i,P) = f(b_i,P)$ for all $1 \leq i \leq \min(m,n)$. So, we must also have $n \leq m$, since if $n > m$, we would have \mbox{$f(r \rightsquigarrow a, P) >_L f(r \rightsquigarrow b, P)$}. Moreover, since $f(r \rightsquigarrow a, P) <_L f(r \rightsquigarrow b, P)$, we must have that $n < m$, for if $n = m$, we would have $f(r \rightsquigarrow a, P) = f(r \rightsquigarrow b, P$), a contradiction. 
We have just shown the existence of some node $b_{n+1}$, for which we must have that \mbox{$f(b_{n+1}, P ) \leq f(a_n, P)$}. Notice that the path $r \rightsquigarrow a$ does not have an $(n+1)^{st}$ node, so it's clear that if $f(b_{n+1}, P) = k$, then $f(P \cup \{a\})_k < f(P \cup \{b\})_k$. Finally, since $n < m$, we have by Property \ref{prob:additive-function}, that $f(a_i, P) \leq f(a_n,P) \leq k$ for all $1 \leq i \leq n$. By an additional application of Property \ref{prob:additive-function} it's easy to see that for all $j > k$, we have $f(P \cup \{a\})_j = f(P \cup \{b\})_j$.  \qed
\end{proof}

From Lemma \ref{lem-path-ineq}, we obtain the following result as an easy Corollary.

\begin{corollary}\label{coro-iff}
Let $r$ be the root of a failure model given by a tree. Given $P \subseteq C$, $a,b \in C - P$. Then $f(r\rightsquigarrow a, P) \lleq f(r\rightsquigarrow b, P)$ if and only if $f(P \cup \{a\}) \lleq f(P \cup \{b\})$.
\end{corollary}

\begin{proof}
Suppose $f(r\rightsquigarrow a, P) \lleq f(r\rightsquigarrow b, P)$. If $f(r\rightsquigarrow a, P) = f(r\rightsquigarrow b, P)$, then since the only nodes which change failure number when considering placements $P$ and $P \cup \{a\}$ are those on the paths $r \rightsquigarrow a$, and each of these nodes' failure numbers increase by $1$, we must have that $f(P \cup \{a\}) = f(P \cup \{b\})$, since the sequence of failure numbers in $r \rightsquigarrow a$ and $r \rightsquigarrow b$ are the same. If $f(r \rightsquigarrow a, P) <_L f(r \rightsquigarrow b, P)$ then by Lemma \ref{lem-path-ineq} the Corollary is proven.

If instead $f(P \cup \{a\}) \lleq f(P\cup \{b\})$, and yet $f(r \rightsquigarrow a, P) >_L f(r \rightsquigarrow b, P)$, then by Lemma \ref{lem-path-ineq} we obtain that $f(P \cup \{a\}) >_L f(P \cup \{b\})$, a contradiction. \qed
\end{proof}

Given a node $u$ in a tree, let $L(u)$ be the set of all leaves which are descendants of $u$.

\begin{lemma}\label{lem-technical}
Given $P \subseteq C$, $a,b \in C$. Let $c$ be the least common ancestor of $a$ and $b$, and let $d$ be the child of $c$ on the path from $c$ to $a$. If \mbox{$f(r\rightsquigarrow a, P) \lleq f(r\rightsquigarrow b, P)$} and $X \subset C - \{a,b\}$ for which $L(d) \cap X = \emptyset$, and $a,b \notin X$ then 
$$f(P \cup X \cup \{a\})  \lleq f(P \cup X \cup \{b\}).$$ 
\end{lemma}

\begin{proof}
We have that $f(r\rightsquigarrow a, P) \lleq f(r\rightsquigarrow b, P)$. Consider $f(r\rightsquigarrow a,  P \cup X)$ and $f(r\rightsquigarrow b, P \cup X)$. We wish to show that $f(r\rightsquigarrow a, P \cup X) \lleq f(r\rightsquigarrow b, P \cup X)$. Since $c$ is the least common ancestor of $a$ and $b$, it is clear that nodes on $r \rightsquigarrow c$ have equivalent failure numbers in both cases. Therefore it suffices to show that $f(c\rightsquigarrow a, P \cup X) \lleq f(c\rightsquigarrow b, P \cup X)$.

Note that since $d \cap L(X) = \emptyset$, we have that $f(c \rightsquigarrow a, P \cup X) = f(c \rightsquigarrow a, P)$. Moreover, since the addition of nodes in $X$ cannot cause failure numbers on the path $c \rightsquigarrow b$ to decrease, we must have that $f(c \rightsquigarrow b, P) \lleq f(c \rightsquigarrow b, P \cup X)$. Altogether, we have that
$$f(c \rightsquigarrow a, P \cup X) = f(c \rightsquigarrow a, P) \lleq f(c \rightsquigarrow b, P) \lleq f(c \rightsquigarrow b, P \cup X).$$
By applying Corollary \ref{coro-iff}, we obtain that $f(P \cup X \cup \{a\}) \lleq f(P \cup X \cup \{b\})$. \qed
\end{proof}

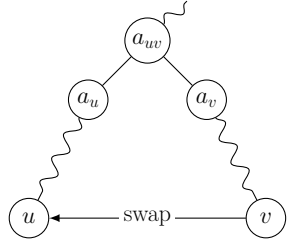
\begin{figure}[h]

		\centering
		\input{proof-figure-thm-1}
		\caption{Named nodes used in Theorem \ref{thm-greedy}. The arrow labeled ``swap" illustrates the leaf nodes between which replicas are moved, and is not an edge of the graph.}\label{fig:thm-1}
\end{figure}

\begin{theorem}\label{thm-greedy}
Let $P_i$ be the partial placement from step $i$ of the greedy algorithm. Then there exists an optimal placement $P^\ast$, with $|P^\ast| = \repfact$ such that $P_i \subseteq P^\ast$.
\end{theorem}

\begin{proof}
The proof proceeds by induction on $i$. $P_0 = \emptyset$ is clearly a subset of any optimal solution. Given $P_i \subseteq P^\ast$ for some optimal solution $P^\ast$, we must show that there is an optimal solution $Q^\ast$ for which $P_{i+1} \subseteq Q^\ast$. Clearly, if $P_{i+1} \subseteq P^\ast$, then we are done, since $P^\ast$ is optimal. In the case where $P_{i+1} \not\subseteq P^\ast$ we must exhibit some optimal solution $Q^\ast$ for which $P_{i+1} \subseteq Q^\ast$. Let $u$ be the leaf which was added to $P_i$ to form $P_{i+1}$. Let $v$ be the leaf in $P^\ast - P_{i+1}$ which has the greatest-depth least common ancestor with $u$, where the depth of a node is given by its distance from the root (see Fig. \ref{fig:thm-1}). We set $Q^\ast = (P^\ast - \{v\}) \cup \{u\}$, and claim that $\vec{f}(Q^\ast) \lleq \vec{f}(P^\ast)$. Since $\vec{f}(P^\ast)$ is optimal, and $P_{i+1} \subseteq Q^\ast$ this will complete our proof.

Clearly, $f(a \rightsquigarrow u, P_i) \lleq f(a \rightsquigarrow v, P_i)$, since otherwise $f(r \rightsquigarrow u, P_i) >_L f(r \rightsquigarrow v, P_i)$, implying that $f(P_i \cup \{u\}) >_L f(P_i \cup \{v\})$, contradicting our use of a greedy algorithm.

Note that $u,v \notin (P^\ast - P_i - \{v\})$. Moreover, by choice of $v$, we have that $L(a) \cap (P^\ast - P_i - \{v\}) = \emptyset$, since the only nodes from $P^\ast$ in $L(a)$ must also be in $P_i$. To complete the proof, we apply Lemma \ref{lem-technical}, setting $X = P^\ast - P_i - \{v\}$. This choice of $X$ is made so as to yield the following equalities. 
$$Q^\ast = (P^\ast - \{v\}) \cup \{u\} = P_i \cup (P^\ast - P_i - \{v\}) \cup \{u\}, $$
$$P^\ast = P_i \cup (P^\ast - P_i - \{v\}) \cup \{v\}. $$

By Lemma \ref{lem-technical}, we obtain inequality in the following formula,
$$f(Q^\ast) = f(P_i \cup (P^\ast - P_i - \{v\}) \cup \{u\}) \lleq f(P_i \cup (P^\ast - P_i - \{v\}) \cup \{v\}) = f(P^\ast).$$
Thereby completing the proof.\qed 
\end{proof}

\section{Balanced Placements}

\begin{figure}[t]
		\begin{minipage}[t]{0.48\textwidth}
		\centering
		\includegraphics[scale=0.09]{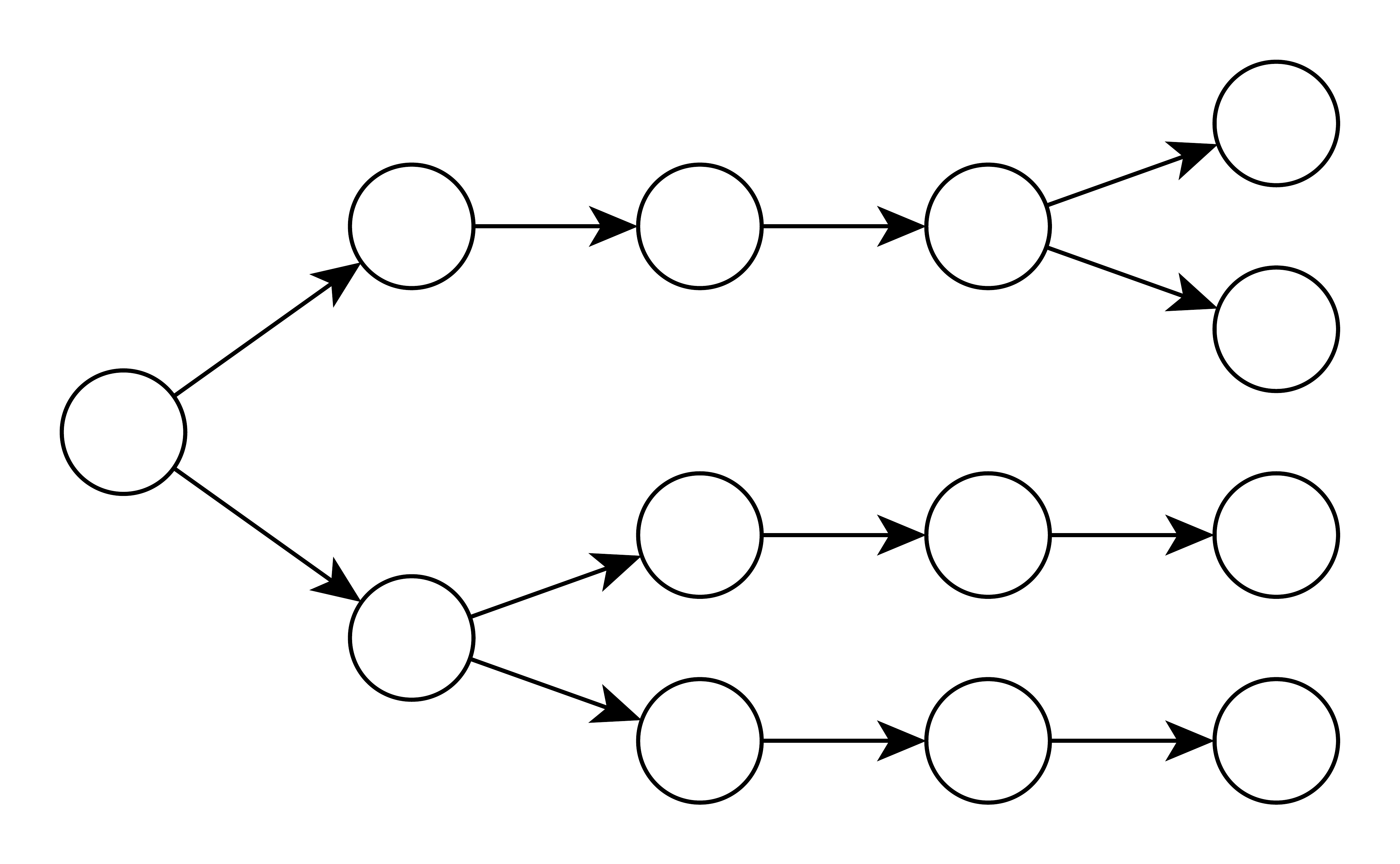}
		\put(0,58){$~~~2$}
		\put(0,40){$1,2$}
		\put(0,23){$1$}
		\put(0,6){$1,2$}
		\caption{Round-robin placement cannot guarantee optimality}\label{fig-cex}
		\end{minipage}
		\begin{minipage}[t]{0.48\textwidth}
		\centering
		\input{proof-figure-thm-2}
		\caption{Nodes used in Theorem \ref {thm-balanced-sufficiency}.} \label{fig-proof-lca}
		\end{minipage}\hfill
\end{figure}

Consider a round-robin placement in which the set of replicas placed at each node is distributed among its children, one replica per child, until all replicas have been placed. This process is then continued recursively at the children. Throughout the process, no child is given more replicas than its subtree has leaf nodes. This method has intuitive appeal, but it does not compute an optimal placement exactly as can be seen from Fig. \ref{fig-cex}. Let placements $P_1$ and $P_2$ consist of the nodes labeled by $1$ and $2$ in Fig. \ref{fig-cex} respectively. Note that both outcomes are round-robin placements. A quick computation reveals that $\vec{f}(P_1) = \langle 1, 1, 7, 0 \rangle \neq \langle 1, 3, 3, 2 \rangle = \vec{f}(P_2)$. Since the placements have different failure aggregates, round-robin placement alone cannot guarantee optimality.

Key to our algorithm is the observation that any placement which lexicominimizes $\vec{f}(P)$ must be \textit{balanced}. If we imagine each child $c_i$ of $u$ as a bin of capacity $\ell_i$, balanced nodes are those in which all unfilled children are approximately ``level'', and no child is filled while children of smaller capacity remain unfilled. These ideas are formalized in the following definitions.

\begin{definition}
Let node $u$ have children indexed $1, ..., k$, and let the subtree rooted at the $i^{th}$ child of node $u$ have $\ell_i$ leaves, and $r_i$ replicas placed on it in placement $P$. A node for which $\ell_i - r_i = 0$ is said to be \emph{filled}. A node for which $\ell_i - r_i > 0$ is said to be \emph{unfilled}.
\end{definition}

\begin{definition}\label{def-balanced}
Node $u$ is said to be \emph{balanced} in placement $P$ iff:
$$	\ell_i - r_i > 0 \implies ~\forall\,j \in \{1,...,k\} ~ (r_i \geq r_j - 1 ) .$$
Placement $P$ is said to be \emph{balanced} if all nodes $v \in V$ are balanced.
\end{definition}

To motivate a proof that lexico-minimum placements must be balanced, consider Fig. \ref{fig-balanced-placement} in which $P_1$ and $P_2$ are sets containing leaf nodes labeled $1$ and $2$ respectively. Fig. \ref{fig-balanced-survivalnums} presents two copies of the same tree, but with failure numbers labeled according to $P_1$ and $P_2$. Upon computing $f(P_1)$ and $f(P_2)$, we find that $f(P_1) = \langle 2, 1, 3, 7 \rangle \lgeq \langle 1, 1, 4, 7 \rangle = f(P_2)$. Note that for placement $P_1$, the root of the tree is unbalanced, therefore $P_1$ is unbalanced. Note also, that $P_2$ is balanced, since each of its nodes are balanced. We invite the reader to verify that $P_2$ is an optimal solution for this tree. 

\begin{table}[tb]
\begin{minipage}[b]{0.33\textwidth}
\centering
\includegraphics[scale=0.075, trim=0 0 0 0, clip=true]{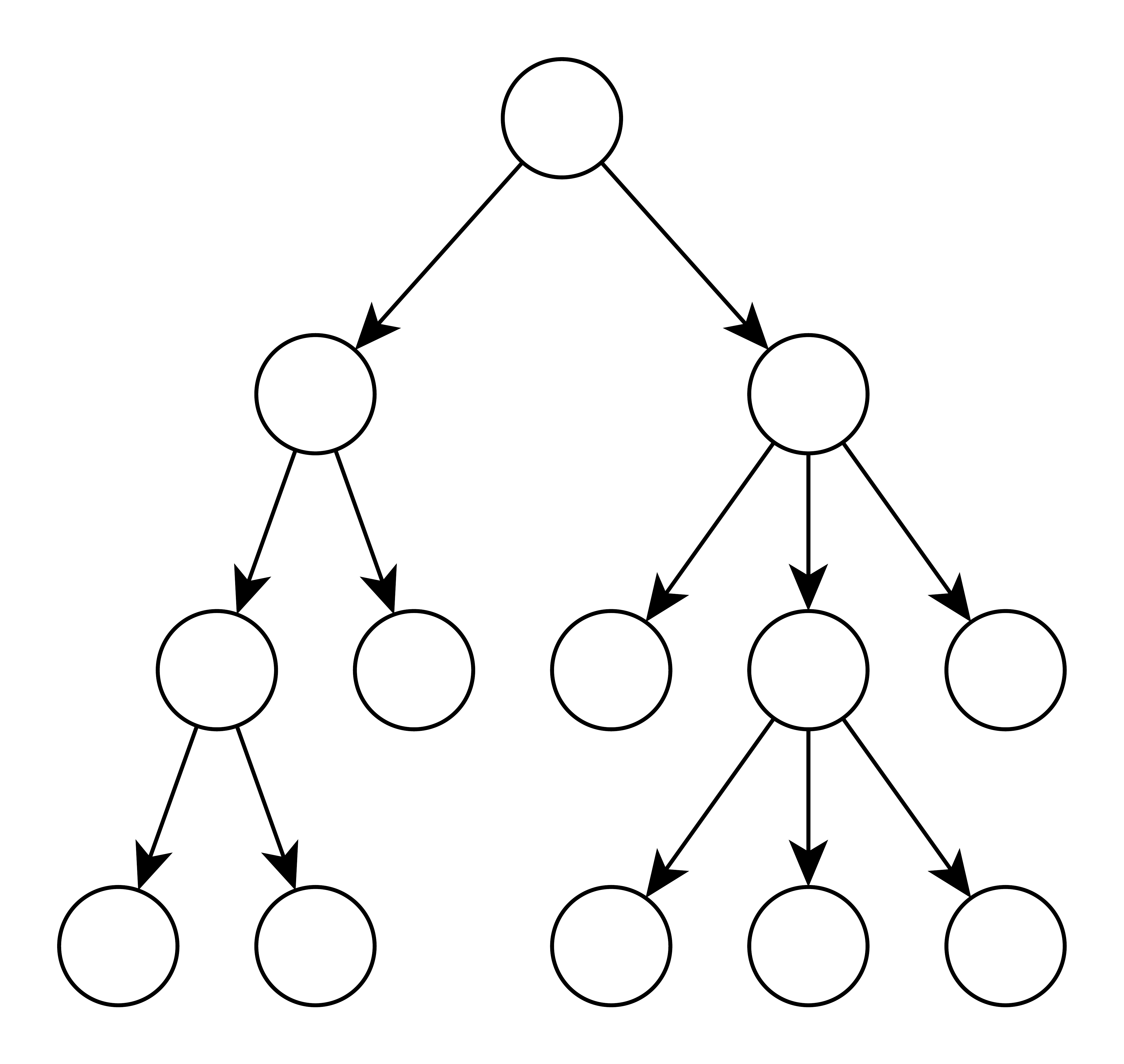}
\put(-88, -3){$1$}
\put(-71, -3){$1$}
\put(-67, 20){$1,2$}
\put(-46, 20){$2$}
\put(-12, 20){$2$}
\vspace{0.25cm}
\captionof{figure}{\strut Placements $P_1, P_2$ \label{fig-balanced-placement}}
\end{minipage}
\begin{minipage}[b]{0.67\textwidth}
\centering
\includegraphics[scale=0.075, trim=0 0 0 0, clip=true]{balanced}
\put(-50.5, 78){$3$} 
\put(-71.5, 54){$3$} 
\put(-29.5, 54){$0$} 
\put(-80, 30.5){$2$} 
\put(-63, 30.5){$1$}
\put(-46, 30.5){$0$}
\put(-29, 30.5){$0$}
\put(-12, 30.5){$0$}
\put(-88.5, 7){$1$} 
\put(-71.5, 7){$1$}
\put(-46, 7){$0$}
\put(-29, 7){$0$}
\put(-12, 7){$0$}
\hspace{0.5cm}
\includegraphics[scale=0.075, trim=0 0 0 0, clip=true]{balanced}
\put(-50.5, 78){$3$} 
\put(-71.5, 54){$1$} 
\put(-29.5, 54){$2$} 
\put(-80, 30.5){$0$} 
\put(-63, 30.5){$1$}
\put(-46, 30.5){$1$}
\put(-29, 30.5){$0$}
\put(-12, 30.5){$1$}
\put(-88.5, 7){$0$} 
\put(-71.5, 7){$0$}
\put(-46, 7){$0$}
\put(-29, 7){$0$}
\put(-12, 7){$0$}
\vspace{0.25cm}
\captionof{figure}{\strut Failure numbers for $P_1$ \textit{(right)} and $P_2$ \textit{(left)}.\label{fig-balanced-survivalnums}}
\end{minipage}
\end{table}

Our main result is that it is \textit{necessary} for an optimal placement to be balanced. However, the balanced property alone is not sufficient to guarantee optimality. To see this, consider the two placements in Fig. \ref{fig-cex}. By definition, both placements are balanced, yet they have different failure aggregates. Therefore, balancing alone is insufficient to guarantee optimality. Despite this, we can use Theorem \ref{thm-balanced-sufficiency} to justify discarding unbalanced solutions as suboptimal. We exploit this property of optimal placements in our algorithm.

\begin{theorem}\label{thm-balanced-sufficiency}
Any placement $P$ in which $\vec{f}(P)$ is lexicominimum among all placements for a given tree must be balanced.
\end{theorem}

\begin{proof}
Suppose $P$ is not balanced, yet $\vec{f}(P)$ is lexicominimum among all placements $P$. We proceed to a contradiction, as follows.

Let $u$ be an unbalanced node in $T$. Let $v$ be an unfilled child of $u$, and let $w$ be a child of $u$ with at least one replica such that $r_v < r_w - 1$. Since $v$ is unfilled, we can take one of the replicas placed on $w$ and place it on $v$. Let $q_w$ be the leaf node from which this replica is taken, and let $q_v$ be the leaf node on which this replica is placed (see Fig. \ref{fig-proof-lca}). Let \mbox{$P^\ast := (P - \{q_w\}) \cup \{q_v\}$}. We aim to show that $P^\ast$ is more optimal than $P$, contradicting $P$ as a lexicominimum.

Let $\vec{f}(P) := \langle p_\repfact, ..., p_0 \rangle$, and $\vec{f}(P^\ast) := \langle p^\ast_\repfact, ..., p^\ast_0 \rangle$. For convenience, we let \mbox{$f(w, P) = m$}. To show that $\vec{f}(P^\ast) <_L \vec{f}(P)$, we aim to prove that $p^\ast_m < p_m$, and that for any $k$ with $\repfact \geq k > m$, that $p^\ast_k = p_k$. We will concentrate on proving the former, and afterwards show that the latter follows easily.

To prove $p^\ast_m < p_m$, observe that as a result of the swap, some nodes change failure number. These nodes all lie on the paths $v \rightsquigarrow q_v$ and $w \rightsquigarrow q_w$. Let $S^-$ (resp. $S^+$) be the set of nodes whose failure numbers change to $m$ (resp. change from $m$), as a result of the swap. Formally, we define
$$S^- := \{x \in V \mid f(x, P)  = m,    f(x, P^\ast) \neq m  \}, $$
$$S^+ := \{x \in V \mid f(x, P) \neq m , f(x, P^\ast) = m \}.$$
By definition, $p^\ast_m = p_m - |S^-| + |S^+|$. We claim that $|S^-| \geq 1$ and $|S^+| = 0$, which yields $p^\ast_m < p_m$. To show $|S^-| \geq 1$, note that $f(w, P) = m$ by definition, and after the swap, the failure number of $w$ changes. Therefore, $|S^-| \geq 1$.

To show $|S^+| = 0$, we must prove that no node whose failure number is affected by the swap has failure number $m$ after the swap has occured. We choose to show a stronger result, that all such node's failure number must be strictly less than $m$. Let $s_v$ be an arbitrary node on the path $v \rightsquigarrow q_v$, and consider the failure number of $s_v$. As a result of the swap, one more replica is counted as failed in each node on this path, therefore $f(s_v, P^\ast) = f(s_v, P) + 1$. Likewise, let $s_w$ be an arbitrary node on path $w \rightsquigarrow q_w$. One less replica is counted as failed in each node on this path, so $f(s_w, P^\ast) = f(s_w, P) - 1$. We will show that $f(s_w, P^\ast) < m$, and $f(s_v, P^\ast) < m$.

First, note that for any $s_w$, by Property \ref{lem-desc} $f(s_w, P^\ast) \leq f(w, P^\ast) = m-1 < m$. Therefore, $f(s_w, P^\ast) < m$, as desired.

To show $f(s_v, P^\ast) < m$, note that by supposition $r_w - 1 > r_v$, and from this we immediately obtain $f(w, P) - 1 > f(v,P)$ by the definition of failure number. Now consider the nodes $s_v$, for which 
$$f(s_v, P) \leq f(v,P) < f(w,P) - 1 = m - 1 \implies  f(s_v, P^\ast) - 1 < m - 1,$$
Where the first inequality is an application of Property \ref{lem-desc}, and the implication follows by substitution. Therefore $f(s_v, P^\ast) < m$ as desired.

Therefore, among all nodes in $P^\ast$ whose failure numbers change as a result of the swap, no node has failure number $m$, so $|S^+| = 0$ as claimed. Moreover, since $f(s, P^\ast) < m$ for any node $s$ whose failure number changes as a result of the swap, we also have proven that $p_k = p^\ast_k$ for all $k$ where $\repfact \geq k > m$. This completes the proof. \qed
\end{proof}

\section{An $O(n\repfact)$ Algorithm}
Our algorithm considers only placements which are balanced. To place $\repfact$ replicas, we start by placing $\repfact$ replicas at the root of the tree, and then proceed to assign these replicas to children of the root. We then recursively carry out the same procedure on each of the children. 

Before the recursive procedure begins, we obtain values of $\ell_i$ at each node by running breadth-first search as a preprocessing phase.
The recursive procedure is then executed in two consecutive phases.
During the \textit{divide} phase, the algorithm is tasked with allocating $r(u)$ replicas placed on node $u$ to the children of $u$. After the divide phase, some child nodes are filled, while others remain unfilled. To achieve balance, each unfilled child $c_i$ will have either $r(c_i)$ or $r(c_i) - 1$ replicas placed upon them. The value of $r(c_i)$ is computed for each $c_i$ as part of the divide phase. The algorithm is then recursively called on each unfilled node to obtain values of optimal solutions for their subtrees. Nodes which are filled require no further processing. The output of this call is a pair of two optimal failure aggregates, one supposing $r(c_i)$ replicas are placed at $c_i$, the other supposing $r(c_i) -1$ are placed. Given these failure aggregates obtained from each child, the \textit{conquer} phase then chooses whether to place $r(c_i)$ or $r(c_i) - 1$ replicas on each unfilled child so as to achieve a lexicominimum failure aggregate for node $u$ overall. For ease of exposition, we describe an $O(n\repfact)$ version of our algorithm in this section, and prove it correct. In Section \ref{sec-improvements} then discuss improvements which can be used to obtain an $O(n + \repfact^2)$ algorithm. Finally, we describe some tree transformations which can be used to obtain an $O(n + \repfact \log \repfact)$ algorithm in Section \ref{sec-best}.

\subsection{Divide Phase}\label{sec-divide}

When node $u$ is first considered, it receives at most two possible values for the number of replicas it could be asked to accommodate. Let these be the values $r(u)$ and $r(u) - 1$. Let $u$ have a list of children indexed $1,2,..., m$, with leaf capacities $\ell_i$ where $1 \leq i \leq m$. The divide phase determines which children will be filled and which will be unfilled. Filled children will have $\ell_i$ replicas placed on them in the optimal solution, while the number of replicas on the unfilled children is determined during the conquer phase.

The set of unfilled children can be determined (without sorting) in an iterative manner using an $O(m)$ time algorithm similar to that for the Fractional Knapsack problem. The main idea of the algorithm is as follows: in each iteration, at least one-half of the children whose status is currently unknown are assigned a filled/unfilled status. To determine which half, the median capacity child (with capacity $\ell_{med}$) is found using the selection algorithm. Based upon the number of replicas that have not been assigned to the filled nodes, either \begin{inparaenum}[a)]\item the set of children $c_i$ with $\ell_i \geq \ell_{med}$ are labeled as ``unfilled" or \item the set of children $c_i$ with $\ell_i \leq \ell_{med}$ are labeled as ``filled"\end{inparaenum}. The algorithm recurses on the remaining unlabeled children. Pseudocode for this algorithm can be found in Algorithm \ref{alg-get-filled}

We briefly sketch the correctness of Algorithm 1. The following invariant holds after every execution of the while loop: 
\begin{equation*}
\max(F)\cdot(|U| + |M|) < r - \sum_{c_i \in F} \ell_i \leq \min(U) \cdot |U| + \sum_{c_i \in M} \ell_i.
\end{equation*}
When $U = \emptyset$ or $F = \emptyset$ the invariant is not well-defined. These conditions are easy to test for: $U = \emptyset$ if and only if $\sum \ell_i = r(u)$, and $F = \emptyset$ if and only if $\ell_i > \floorfrac{r(u)}{|M|}$ for all $i$. Hence in what follows, we will work only with cases where $U \neq \emptyset$ and $F \neq \emptyset$. At the end of the algorithm, $M = \emptyset$, and the invariant reduces to the following
\begin{equation}\label{eqn-invariant-reduced}
\max(F) < \frac{r - \sum_{c_i \in F} \ell_i}{|U|} \leq \min(U).
\end{equation}
Equation \ref{eqn-invariant-reduced} indicates that the average number of replicas placed on the unfilled nodes lies between the maximum value of $F$ and the minimum value of $U$. From this, it is easy to see that the labeling is correct. Suppose that some filled child $c_i \in F$ has been incorrectly classified. This child contains at most $\ell_i - 1$ replicas, and yet is still unfilled. Moreover, to attain the average, some unfilled child must be assigned at least $\ceilfrac{r - \sum_{c_i \in F}^{\ell_i}}{|U|}$ replicas. Taking the difference of the number of replicas assigned to these two unfilled nodes, we have
\begin{align*}
 & \Big\lceil\dfrac{r - \sum_{c_i \in F}\ell_i}{|U|}\Big\rceil - \ell_i + 1 \\
 >~~ & \Big\lceil\dfrac{r - \sum_{c_i \in F}\ell_i}{|U|}\Big\rceil - \max(F) + 1 \\
 \geq~~ & \Big\lceil\dfrac{r - \sum_{c_i \in F}\ell_i}{|U|}\Big\rceil - \max(F) + 2 \geq 2
\end{align*} 
which is a violation of the balanced placement property. Therefore, all replicas are correctly classified. This completes the proof sketch.

\begin{algorithm}[t]
\SetKwProg{Fn}{Function}{begin}{end}
\Fn{\getFilled{$M$, $r$}}{
	$F \gets \emptyset$ ; $U \gets \emptyset$ \tcp*[r]{$F$ := filled children $U$ := unfilled children}
	\While{$M \neq \emptyset$}{
		$\ell_{med} \gets \text{ median capacity of children in } M $ \;
		$M_1 \gets \{c_i \in M \mid \ell_i < \ell_{med}  \} $ \;
		$M_2 \gets \{c_i \in M \mid \ell_i = \ell_{med}  \} $ \;
		$M_3 \gets \{c_i \in M \mid \ell_i > \ell_{med}  \} $ \;
		$x \gets r - \sum_{c_i \in F \cup M_1 \cup M_2} \ell_i$ \tcp*[r]{$x$ to be distributed among $M_3 \cup U$} 
		\uIf(\tcp*[f]{$M_1 \cup M_2$ guaranteed filled}){$x \geq \ell_{med} \cdot (|U| + |M_3|)$}{
			$F \gets F\cup M_1 \cup M_2$ \; $M \gets M - (M_1 \cup M_2)$ \;
		}\Else(\tcp*[f]{$M_2 \cup M_3$ guaranteed unfilled}) {
			$U \gets F\cup M_2 \cup M_3$ \; $M \gets M - (M_2 \cup M_3)$ \;
		}
	}
	\Return{($F$, $U$)} \tcp*[r]{return filled and unfilled children}
}
\caption{Determines filled and unfilled nodes}\label{alg-get-filled}
\end{algorithm}

Suppose we know that we only need to find placements of size $r(u)$ and $r(u) - 1$ for node $u$. Moreover, we know that in an optimal placement of size $r(u)$, each child $c_i$ only needs to accomodate either $r(c_i)$ or $r(c_i) - 1$ replicas. Suppose that optimal placements of size $r(c_i)$ and $r(c_i) - 1$ are available at each child $c_i$. Theorem \ref{thm-two-values} shows that these placements are all that is required to compute optimal placements of size $r(u)$ \emph{and also of size} $r(u) - 1$.

\begin{theorem}\label{thm-two-values}
In any case where $r(u)$ or $r(u) - 1$ replicas must be balanced among $k$ unfilled children, it suffices to consider placing either $\ceilfrac{r(u) - L}{k}$ or $\floorfrac{r(u)- L -1}{k}$ replicas at each unfilled child.
\end{theorem}
\begin{proof}
Let $s := r(u) - L$. Suppose $s \bmod k = 0$. If $s$ replicas are placed at $u$, then all unfilled children receive exactly $\frac{s}{k} ~(= \ceilfrac{s}{k})$ replicas. If $s - 1$ replicas are placed at $u$, one child gets $\frac{s}{k} - 1 = \floorfrac{s - 1}{k}$ replicas. If instead $s \bmod k > 0$, then the average number of replicas on each unfilled child is $\frac{s}{k} \notin \ints$. To attain this average using integer values, values both above and below $\frac{s}{k}$ are needed. However, since the unfilled children must be balanced, whatever values selected must have absolute difference at most 1. The only two integer values satisfying these requirements are $\ceilfrac{s}{k}$ and $\floorfrac{s}{k}$. But $\floorfrac{s}{k} = \floorfrac{s - 1}{k}$ when \mbox{$s \bmod k > 0$}. \qed
\end{proof}

\subsection{Conquer Phase}\label{sec-conquer}

Once the recursive call completes, we combine the results from each of the children to achieve the lexicographic minimum overall. Our task in this phase is to select $(r(u) - L) \bmod k$ unfilled children on which $\ceilfrac{r(u)-L}{k}$ replicas will be placed, and place $\floorfrac{r(u) - L - 1}{k}$ replicas on the remaining unfilled children. We need to do this in such a way that the resulting placement is lexicominimum. Recall also that we must return two values, one for $r(u)$ and another for $r(u) - 1$. We show how to obtain a solution in the $r(u) - 1$ case using a greedy algorithm. A solution for $r(u)$ can easily be obtained thereafter.
In this section, when two vectors are compared or summed, we are implicitly making use of an $O(\repfact)$ function for comparing two vectors of length $\repfact$ in lexicographic order.

Let $\vec{a}_i$ (respectively $\vec{b}_i$) represent the lexicominimum value of $\vec{f}(P)$ where $P$ is any placement of $\floorfrac{r(u)- L -1}{k}$ (respectively $\ceilfrac{r(u) - L}{k}$) replicas on child $i$. Recall that $\vec{a}_i, \vec{b}_i \in \natnum^{\repfact + 1}$, and are available as the result of the recursive call. We solve the optimization problem by encoding the decision to take $\vec{b}_i$ over $\vec{a}_i$ as a decision variable $x_i \in \{0,1\}$, for which either $x_i = 0$ if $\vec{a}_i$ is selected, or $x_i = 1$ if $\vec{b}_i$ is selected. The problem can then be described as an assignment of values to $x_i$ according to the following system of constraints, in which all arithmetic operations are performed point-wise.
\begin{equation}\label{eqn-greedy-constraints}
\min \displaystyle\sum_{i} \vec{a}_i + (\vec{b}_i - \vec{a}_i)x_i, ~~~ \text{subj. to: } \displaystyle\sum_{i}x_i = (r(u) - L) \bmod k.
\end{equation}

An assignment of $x_i$ which satisfies the requirements in (\ref{eqn-greedy-constraints}) can be found by computing $\vec{b}_i - \vec{a}_i$ for all $i$, and greedily assigning $x_i = 1$ to those $i$ which have the $(r(u) - L) \bmod k$ smallest values of $\vec{b}_i - \vec{a}_i$. This is formally stated as 

\begin{theorem}\label{thm-greedy-system}
Let $\pi := (\pi_1,\pi_2,...,\pi_k)$ be a permutation of $\{1,2,...,k\}$ such that:
$$\vec{b}_{\pi_1} - \vec{a}_{\pi_1} \lleq \vec{b}_{\pi_2} - \vec{a}_{\pi_2} \lleq ... \lleq \vec{b}_{\pi_k} - \vec{a}_{\pi_k}~.$$
If vector $\vec{x} = \langle x_1, ..., x_k\rangle$ is defined according to the following rules: set $x_{\pi_i} = 1$ iff \mbox{$i < (r(u) - L) \bmod k$}, else $x_{\pi_i} = 0$, then $\vec{x}$ is an optimal solution to (\ref{eqn-greedy-constraints}).
\end{theorem}

The following Lemma greatly simplifies the proof of Theorem \ref{thm-greedy-system}. 

\begin{lemma}\label{lem-logroup}
$\langle\ints^n, +\rangle$ forms a linearly-ordered group under $\lleq$. In particular, for any $\vec{x}, \vec{y}, \vec{z} \in \ints^n, \vec{x} \lleq \vec{y} \implies \vec{x} + \vec{z} \lleq \vec{y} + \vec{z}$.
\end{lemma}
A straight-forward proof of Lemma \ref{lem-logroup} can be found in the appendix.

\begin{proof}[Proof of Theorem \ref{thm-greedy-system}]
First, notice that a solution to (\ref{eqn-greedy-constraints}) which minimizes the quantity $\sum_i (\vec{b}_i - \vec{a}_i) x_i$ also minimizes the quantity $\sum_i \vec{a}_i + (\vec{b}_i - \vec{a}_i)x_i.$ It suffices to minimize the former quantity, which can be done by considering only those values of $(\vec{b}_i - \vec{a}_i)$ for which $x_i = 1$. For convenience, we consider $\vec{x}$ to be the characteristic vector of a set $S \subseteq \{1,...,k\}$. We show that no other set $S'$ can yield a characteristic vector $\vec{x}'$ which is strictly better than $\vec{x}$ as follows.

Let $\alpha := (r(u) - L) \bmod k$, and let $S := \{\pi_1, ..., \pi_{\alpha - 1} \}$ be the first $\alpha - 1$ entries of $\pi$ taken as a set. Suppose that there is some $S'$ which represents a feasible assignment of variables to $\vec{x}'$ for which $\vec{x}'$ is a strictly better solution than $\vec{x}$. $S' \subseteq \{1, ..., k\}$, such that $|S'| = \alpha - 1$, and $S' \neq S$. Since $S' \neq S$, and $|S'| = |S|$ we have that $S - S' \neq \emptyset$ and $S' - S \neq \emptyset$. Let $i \in S-S'$ and $j \in S' - S$. We claim that we can form a better placement, $S^\ast = (S' - \{j\}) \cup \{i\}$. Specifically,
\begin{equation}\label{eqn-claim01}
\sum_{\ell \in S^\ast} (\vec{b}_\ell - \vec{a}_\ell) \lleq \sum_{m \in S'} (\vec{b}_m - \vec{a}_m)~.
\end{equation}
which implies that replacing a single element in $S'$ with one from $S$ does not cause the quantity minimized in (\ref{eqn-greedy-constraints}) to increase.

To prove (\ref{eqn-claim01}) note that \mbox{$j \notin S$ and $i \in S \implies (\vec{b}_i - \vec{a}_i) \lleq (\vec{b}_j - \vec{a}_j)$.} We now apply Lemma \ref{lem-logroup}, setting $\vec{x} = (\vec{b}_i - \vec{a}_i)$, $\vec{y} = (\vec{b}_j - \vec{a}_j)$, and \mbox{$\vec{z} = \sum_{\ell \in (S^\ast - \{i\})} (\vec{b}_\ell - \vec{a}_\ell)$}. This yields
$$\sum_{\ell \in (S^\ast - \{i\})} (\vec{b}_\ell - \vec{a}_\ell) + (\vec{b}_i - \vec{a}_i) \lleq 
\sum_{\ell \in (S^\ast - \{i\})} (\vec{b}_\ell - \vec{a}_\ell) + (\vec{b}_j - \vec{a}_j)~.$$
But since $S^\ast - \{i\} = S' - \{j\}$, we have that
\begin{equation}\label{eqn-claim02}
\sum_{\ell \in (S^\ast - \{i\})} (\vec{b}_\ell - \vec{a}_\ell) + (\vec{b}_i - \vec{a}_i) \lleq 
\sum_{m \in (S' - \{j\})} (\vec{b}_m - \vec{a}_m) + (\vec{b}_j - \vec{a}_j)~.
\end{equation}
Clearly, (\ref{eqn-claim02}) $\implies$ (\ref{eqn-claim01}), thereby proving (\ref{eqn-claim01}).
 This shows that any solution which is not $S$ can be modified to swap in one extra member of $S$ without increasing the quantity minimized in (\ref{eqn-greedy-constraints}). By induction, it is possible to include every element from $S$, until $S$ itself is reached. Therefore, $\vec{x}$ is an optimal solution to (\ref{eqn-greedy-constraints}). \qed
\end{proof}

In the algorithm, we find an optimal solution to (\ref{eqn-greedy-constraints}) by assigning $\ceilfrac{r(u) - L - 1}{k}$ replicas to those children where $i$ is such that \mbox{$1 \leq i < (r(u) - L) \bmod k$}, and $\floorfrac{r(u) - L}{k}$ replicas to those remaining. To do this, we find the unfilled child having the $((r(u) - L) \bmod k)^{th}$ largest value of $\vec{b}_i - \vec{a}_i$ using linear-time selection, and use the partition procedure from quicksort to find those children having values below the selected child. This takes time $O(k\repfact)$ at each node.

At the end of the conquer phase, we compute and return the sum\footnote{In the mentioned sum we assume for notational convenience, that the vectors have been indexed in increasing order of $\vec{b}_i - \vec{a}_i$, although the algorithm performs no such sorting.}
\begin{equation}\label{eqn-recursive-sum}
\sum_{i \,<\, (r(u) - L )\bmod k}\vec{b}_i + \sum_{i \,\geq\, (r(u) - L) \bmod k} \vec{a}_i  + \sum_{j \,:\, \text{filled}} \vec{f}(P_j) + \vec{1}_{r(u)-1},
\end{equation} where $P_j$ is the placement of $\ell_j$ replicas on child $j$ and $\vec{1}_{r(u)-1}$ is a vector of length $\rho$ having a one in entry $r(u)-1$ and zeroes everywhere else. The term $\vec{1}_{r(u)-1}$ accounts for the failure number of $u$. This sum gives the value of an optimal placement of size $r(u) - 1$. Note there are $k+1$ terms in the sum, each of which is a vector of length at most $\rho + 1$. Both computing the sum and performing the selection take $O(k\repfact)$ time at each node, yielding $O(n\repfact)$ time overall.

We have only focused upon computing the \textit{value} of the optimal solution. The solution itself can be recovered easily by storing the decisions made during the conquer phase at each node, and then combining them to output an optimal placement.

\section{An $O(n + \repfact^2)$ Algorithm}
\label{sec-improvements}

An $O(n + \repfact^2)$ running time can be achieved by an $O(n)$ divide phase, and an $O(\repfact^2)$ conquer phase. The divide phase already takes at most $O(n)$ time overall, so to achieve our goal, we concern ourselves with optimizing the conquer phase. The conquer phase can be improved upon by making two changes. First, we modify the vector representation used for return values. Second, we transform the structure of the tree to avoid pathological cases. 

In the remainder of the paper, we will use array notation to refer to entries of vectors. For a vector $\vec{v}$, the $k^{th}$ entry of $\vec{v}$ is denoted $\vec{v}[k]$.

\subsubsection{Compact Vector Representation}
Observe that the maximum failure number returned from child $c_i$ is $r(c_i)$. This along with Property \ref{lem-desc} implies that the vector returned from $c_i$ will have a zero in indices $\repfact, \repfact-1, ..., r(c_i) +1$. To avoid wasting space, we modify the algorithm to return vectors of length only $r(c_i)$.  At each node, we then compute (\ref{eqn-recursive-sum}) by summing entries in increasing order of their index. Specifically, to compute $\vec{v}_1 + \vec{v}_2 + ... + \vec{v}_k$, where each vector $\vec{v}_j$ has length $r(c_i)$, we first allocate an empty vector $\vec{w}$, of size $r(c_i)$, to store the result of the sum. Then, for each vector $\vec{v}_j$, we set $\vec{w}[i] \gets \vec{w}[i] + \vec{v}_j[i]$ for indices $i$ from $0$ up to $r(c_i)$. After all vectors have been processed, $\vec{w} = \vec{v}_1 + ... + \vec{v}_k$. This algorithm takes \mbox{$r(c_1) + ... + r(c_k) = O(r(u))$} time. 
Using smaller vectors also implies that the $((r(u) - L) \bmod k)^{th}$ best child is found in $O(r(u))$ time, since each unfilled child returns a vector of size at most $O(\frac{r(u)}{k})$, and there are only $k$ unfilled children to compare. 
With these modifications the conquer phase takes $O(r(u))$ time at node $u$.

\subsubsection{Tree Transformations}
Note that for each $i$, nodes at depth $i$ have $O(\repfact)$ replicas placed on them in total. We can therefore achieve an $O(\repfact^2)$ time conquer phase overall by ensuring that the conquer phase only needs to occur in at most $O(\repfact)$ levels of the tree. To do this, we observe that when $r(u) = 1$, any leaf with minimum depth forms an optimal placement. Recursive calls can therefore be stopped once $r(u) = 1$.
To ensure that $r(u) = 1$ after $O(\repfact)$ levels, we contract paths on which all nodes have degree two into a single pseudonode during the preprocessing phase. The length of this contracted path is stored in the pseudonode, and is accounted for when computing the sum. This suffices to ensure $r(u)$ decreases by at least one at each level, yielding an $O(n + \repfact^2)$ algorithm.

\section{An $O(n + \repfact \log \repfact)$ Algorithm}
\label{sec-best}
In this section, we extend ideas about tree transformation from the last section to develop an algorithm in which the conquer phase only needs to occur in at most $O(\log \repfact)$ levels. We achieve this by refining the tree transformations described in Section \ref{sec-improvements}.

To ensure that there are only $O(\log \repfact)$ levels in the tree, we transform the tree so as to guarantee that as the conquer phase proceeds down the tree, $r(u)$ decreases by at least a factor of two at each level. This happens automatically when there are two or more unfilled nodes at each node, since to balance the unfilled children, at most $\ceilfrac{r(u) - L}{2}$ replicas will be placed on each of them. Problems can therefore only arise when a tree has a path of nodes each of which have a single, unfilled child. We call such a path a \textit{degenerate chain}. By detecting and contracting all such degenerate chains, we can achieve an $O(\repfact \log \repfact)$ conquer phase.

Fig. \ref{fig-degenerate-unfilled-case} illustrates a degenerate chain. In this figure, each $T_i$ with $1 \leq i \leq t - 1$ is the set of all descendant nodes of $v_i$ which are filled. Thus, $v_1, ..., v_{t-1}$ each have only a single unfilled child (since each $v_i$ has $v_{i+1}$ as an child). In contrast, node $v_t$ has at least two unfilled children. It is easy to see that if the number of leaves in each $T_i$ is $O(1)$ then $t$, the length of the chain, can be as large as $O(\repfact)$. This would imply that there can be $O(\repfact)$ levels in the tree where the entire conquer phase is required. To remove degenerate chains, we contract nodes $v_1, ..., v_{t-1}$ into a single pseudonode $w$, as in Fig. \ref{fig-contracted-nodes}. However, we must take care to ensure that the pair of vectors which pseudonode $w$ returns takes into account contributions from the entire contracted structure. We will continue to use $v_i$ and $T_i$ throughout the remainder of this section to refer to nodes in a degenerate chain.

To find and contract degenerate chains, we add an additional phase, the \textit{transform} phase, which takes place between the divide and conquer phases. Recall that after the divide phase, the set of filled and unfilled children are available at each node. Finding nodes in a degenerate chain is therefore easily done via a breadth-first search. We next consider what information must be stored in the pseudonode, to ensure that correct results are maintained.

\begin{figure}
	\begin{subfigure}{0.60\textwidth}
	\centering
	\input{degenerate-case}
	\caption{A degenerate chain.}\label{fig-degenerate-unfilled-case}
	\end{subfigure}
	\begin{subfigure}{0.35\textwidth}
	\centering
	\raisebox{1cm}{ \input{contracted-nodes} }
	\caption{Contracted pseudonode.}\label{fig-contracted-nodes}
	\end{subfigure}
	\caption{Illustration of a degenerate chain in which each $v_i$ where $1 \leq i \leq t-1$ represents a node which has a single unfilled child. All filled descendents of node $v_i$ are collectively represented as $T_i$. In the figure on the right, nodes $v_1, ..., v_{t-1}$ have been contracted into pseudonode $w$.}
\end{figure}
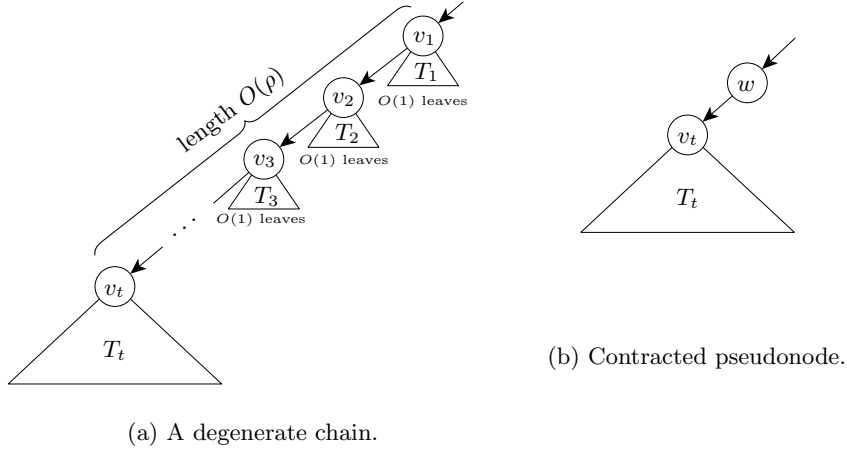

Let $(\vec{a_w}, \vec{b_w})$ be the pair of values which will be returned by pseudonode $w$ at the end of the conquer phase. In order for the transformation to be correct, the vectors $(\vec{a_w}, \vec{b_w})$ must be the same as those which would have been returned at node $v_1$ had no transformation occurred. To ensure this, we must consider and include the contribution of each node in the set \mbox{$T_1 \cup ... \cup T_{t-1} \cup \{v_1, ..., v_{t-1}\}$}. It is easy to see that the failure numbers of nodes in $\{v_1, ..., v_{t-1}\}$ depend only upon whether $r(v_t)$ or $r(v_t) - 1$ replicas are placed on node $v_t$, while the filled nodes in sets $T_1, ..., T_{t-1}$ have no such dependency. Observe that if $r(v_t)$ replicas are placed on $v_t$, then $r(v_i)$ replicas are placed at each node $v_i$. If instead $r(v_t) - 1$ replicas are placed, then $r(v_i) - 1$ replicas are placed at each $v_i$. Since values of $r(v_i)$ are available at each node after the divide phase, enough information is present to contract the degenerate chain before the conquer phase is performed.

The remainder of this section focuses on the technical details needed to support our claim that the transform phase can be implemented in time \mbox{$O(n + \repfact \log \repfact)$} overall. Let $S_w := T_1 \cup ... \cup T_{t-1} \cup \{v_1, ..., v_{t-1}\}$, and let the contibution of nodes in $S_w$ to $\vec{a_w}$ and $\vec{b_w}$ be given by vectors $\vec{a}$ and $\vec{b}$ respectively. The transform phase is then tasked with computing $\vec{a}$ and $\vec{b}$, and contracting the degenerate chain. We will show that this can be done in time $O(|S_w| + r(v_1))$ for each pseudonode $w$.

Pseudocode for the transform phase is given in Algorithm \ref{alg-transf}. The transform phase is started at the root of the tree by invoking \transf{$root,false, \repfact$}. \transf is a modified recursive breadth-first search. As the recursion proceeds down the tree, each node is tested to see if it is part of a degenerate chain (lines \ref{line-bottom} and \ref{line-test2}). If a node is not part of a degenerate chain, the call continues on all unfilled children (line \ref{line-pass-on}). The first node ($v_1$) in a degenerate chain is marked by passing down $chain \gets true$ at lines \ref{line-mark-ru} and \ref{line-mark-rv1}. The value of $r(v_1)$ is also passed down to the bottom of the chain at lines \ref{line-mark-ru} and \ref{line-mark-rv1}. Once the bottom of the chain (node $v_t$) has been reached, the algorithm allocates memory for three vectors, $\vec{a}, \vec{b}$ and $\vec{f}$, each of size $r(v_1)+1$ (line \ref{line-alloc}). These vectors are then passed up through the entire degenerate chain (line \ref{line-return}), along with node $u$, whose use will be explained later. When a node $u$ in a degenerate chain receives $\vec{a}, \vec{b}$, and $\vec{f}$, $u$ adds its contribution to each vector (lines \ref{line-contribstart}-\ref{line-contribend}). The contribution of node $u$ consists of two parts. First, the contribution of the filled nodes is added to $\vec{f}$ by invoking a special \filled subroutine (see Algorithm \ref{alg-filled}) which computes the sum of the failure aggregates of each filled child of $u$ (lines \ref{line-contribstart}-\ref{line-filledend}). Note that \filled uses pass-by-reference semantics when passing in the value of $\vec{f}$. The contribution of node $u$ itself is then added, by summing the number of leaves in all of the filled children, and the number of replicas on the single unfilled child, $v$ (lines \ref{line-ustart}-\ref{line-contribend}). By the time that the recursion reaches the start of the chain on the way back up (line \ref{line-chainback-to-start}), all nodes have added their contribution, and the pseudonode is created and returned (line \ref{line-pseudonodecreate}).

\begin{algorithm}[t]

\SetKwFunction{transf}{Transform}\SetKwFunction{mkPseudo}{Make-Pseudonode}
\SetKwProg{Fn}{Function}{begin}{end}

\Fn{\transf{$u, chain, r(v_1)$}}{
	\If{$u$ has two or more unfilled children}{ \label{line-bottom}
		\ForEach{child $c_i$ unfilled}{
			\label{line-pass-on}$(-, -, -, x) \gets $\transf{$c_i, false, \bot$} \;
			$c_i \gets x$ \label{line-update}\; 
		}
		\lIf{$chain = false$} { \Return{$(\bot, \bot, \bot, u)$} \label{line-vt}} 
		\lElse(\tcp*[f]{$3\cdot O(r(v_1))$ time}){ \Return{$(\vec{0}_{r(v_1) + 1}, \vec{0}_{r(v_1)+1},  \vec{0}_{r(v_1)+1}, u)$} \label{line-alloc} } 
	}
	\If{$u$ has one unfilled child, $v$}{ \label{line-test2}
		\If{$chain = false$} {
			\tcp{pass $r(v)$ as max vector length}$(\vec{a},\vec{b},\vec{f},x) \gets$ \transf{$v, true, r(v)$} \label{line-mark-ru}\;
		}\Else{
			$(\vec{a},\vec{b},\vec{f},x) \gets$ \transf{$v, true, r(v_1)$} \label{line-mark-rv1}\;
		}
		\ForEach{filled child $c_i$} { \label{line-contribstart}
			\filled{$c_i, \vec{f}$} \tcp*[r]{$O(n_i)$ time} \label{line-filledend}
		}
		$k \gets \sum_i \ell_i + r(v) - 1$\label{line-ustart}\;
		$\vec{a}[k+1] \gets \vec{a}[k+1] + 1$\;
		$\vec{b}[k] \gets \vec{b}[k] + 1$\label{line-contribend}\;
		\If{$chain = false$}{ \label{line-chainback-to-start}
			$x \gets $ \mkPseudo{$\vec{a}, \vec{b}, \vec{f}, x$} \label{line-pseudonodecreate}
		}
		\Return{$(\vec{a},\vec{b},\vec{f},x)$}\label{line-return}
	}
}
\caption{Transform phase}\label{alg-transf}
\end{algorithm}

The transformation takes place as \transf is returned back up the tree. At the end of the degenerate chain, node $v_t$ is returned (lines \ref{line-vt}-\ref{line-alloc}), and this value is passed along the length of the entire chain (line \ref{line-return}), until reaching the beginning of the chain, where the pseudonode is created and returned (line \ref{line-pseudonodecreate}). When the beginning of the chain is reached, the parent of $v_1$ updates its reference (line \ref{line-update}) to refer to the newly created pseudonode. At line \ref{line-update} note that if $c_i$ was \textit{not} the beginning of a degenerate chain, $x = c_i$ and the assignment has no effect (see lines \ref{line-vt}-\ref{line-alloc}).

We provide pseudocode for the  \filled and \mkPseudo subroutines in Algorithms \ref{alg-filled} and \ref{alg-mkpseudo}. The \mkPseudo subroutine runs in $O(1)$ time. It is easy to see that the \filled routine runs in $O(n_i)$ time, where $n_i$ is the number of nodes in the subtree rooted at child $c_i$. The \transf routine therefore takes $O(|T_i|)$ time to process a single node $v_i$. The time needed for \transf to process an entire degenerate chain is therefore $O(|S_w|) + 3\cdot O(r(v_1))$, where the $3\cdot O(r(v_1))$ term arises from allocating memory for vectors $\vec{a}$, $\vec{b}$ and $\vec{f}$ at the last node of the chain.

When we sum this time over all degenerate chains, we obtain a running time of $O(n + \repfact \log \repfact)$ for the transform phase. To reach this result, we examine the sum of $r(v_1)$ for all pseudonodes at level $i$. Since there are at most $\repfact$ replicas at each level $i$, this sum can be at most $O(\repfact)$ in any level. There are only $O(\log \repfact)$ levels where $r(u) > 1$ after degenerate chains have been contracted, thus, pseudonodes can be only be present in the first $O(\log \repfact)$ levels of the tree. Therefore the $3\cdot O(r(v_1))$ term sums to $O(\repfact \log \repfact)$ overall. Since $|S_w|$ clearly sums to $O(n)$ overall, the transform phase takes at most $O(n + \repfact \log \repfact)$ time.

Finally, after the transformation has completed, we can ensure that the value of $r(u)$ decreases by a factor of two at each level. This implies that there are only $O(\log \repfact)$ levels where the conquer phase needs to be run in its entirety. Therefore, the conquer phase takes $O(\repfact \log \repfact)$ time overall. When combined with the $O(n)$ divide phase and the $O(n + \repfact \log \repfact$) transform phase, this yields an $O(n + \repfact \log \repfact)$ algorithm for solving replica placement in a tree.

\begin{algorithm}[t]
\SetKwProg{Fn}{Function}{begin}{end}
\Fn{\filled{$u$, $\vec{f}$}}{
	\ElseIf{$u$ is a leaf}{
		$\vec{f}[0] \gets \vec{f}[0]  + 1$\;
		\Return \;
	}{ 
		\ForEach{child $c_i$}{
			\filled{$c_i, \vec{f}$}
		}
		$a \gets \sum_i \ell_i$ \;
		$\vec{f}[a] \gets \vec{f}[a] + 1$\;
		\Return \;
	}
}
\caption{Computes failure aggregate of filled nodes}\label{alg-filled}
\end{algorithm}

\begin{algorithm}[t]
\SetKwProg{Fn}{Function}{begin}{end}
\Fn{\mkPseudo{$\vec{a}$, $\vec{b}$, $\vec{f}$, $x$}}{
	allocate a new node $node$\;
	$node.\vec{a} \gets \vec{a} + \vec{f}$\;
	$node.\vec{b} \gets \vec{b} + \vec{f}$\;
	$node.child \gets x$\;
	\Return $node$
}
\caption{Creates and returns a new pseudonode}\label{alg-mkpseudo}
\end{algorithm}

\section{Conclusion}\label{sec-conclusion}

In this paper, we formulate the replica placement problem and show that it can be solved by a greedy algorithm in $O(n^2 \repfact)$ time. In search of a faster algorithm, we prove that any optimal placement in a tree must be balanced. We then exploit this property to give a $O(n\repfact)$ algorithm for finding such an optimal placement. The running time of this algorithm is then improved, yielding an $O(n + \repfact \log \repfact)$ algorithm. An interesting next step would consist of proving a lower bound for this problem, and seeing how our algorithm compares. In future work we plan to consider replica placement on additional classes of graphs, such as special cases of bipartite graphs.

We would like to acknowledge insightful comments from S. Venkatesan and Balaji Raghavachari during meetings about results contained in this paper, as well as comments from Conner Davis on a draft version of this paper.

%
\bibliography{wads-2015}
\bibliographystyle{splncs03}

\section*{Appendix}

\subsection*{Proof of Property \ref{lem-desc}}
The following property from Section \ref{sec:model} is an easy result regarding failure numbers which is used in the proofs of Theorems \ref{thm-greedy} and \ref{thm-balanced-sufficiency}.

\begin{proof}[Proof of Property \ref{lem-desc}]
Suppose that in $P$ there are $k_i$ replicas placed on leaves in the subtree rooted at $i$. If $i$ fails, $k_i$ replicas fail, yielding $s(i,P) = \repfact - k_i$. Let $k_j$ be the number of replicas in the subtree rooted at $j$. Clearly, $k_j \leq k_i$, yielding the result. \qed
\end{proof}

\subsection*{NP-Hardness of Problem \ref{prob:additive-function}}

However, Problem \ref{prob:additive-function} is NP-hard to solve exactly or approximately. In particular, we can reduce the well-known problems of \textsc{Independent Set} (IS) and \textsc{Dominating Set} (DS)  to Problem \ref{prob:additive-function}. The reduction from DS shows that minimizing only the first entry of $f(P)$, $f_r$, is NP-hard, while the reduction from IS shows that lexico-minimizing the vector \emph{down to} the second-to-last entry, $f_1$ is NP-hard. That is to say, if it were possible to lexico-minimize the vector $\langle f_r, ..., f_2, f_1 \rangle$ in polynomial time, it would imply P = NP.

\subsection*{Proof of Lemma \ref{lem-logroup}}

The proof of Theorem \ref{thm-greedy-system} is greatly simplified through use of an algebraic property of addition on $\ints^n$ under lexicographic order. Recall that a \textit{group} is a pair $\langle S, \cdot \rangle$, where $S$ is a set, and $\cdot$ is a binary operation which is 
\begin{inparaenum}[1)]
\item closed for $S$, 
\item is associative, and has both 
\item an identity and 
\item inverses.
\end{inparaenum}
A \textit{linearly-ordered group} is a group $G = \langle S, \cdot \rangle$, along with a linear-order $\leq$ on $S$ in which for all $x,y,z \in S$, $x \leq y \implies x \cdot z \leq y \cdot z$, i.e. the linear-order on $G$ is \textit{translation-invariant}. Lemma \ref{lem-logroup} states that $\ints^n$ under $\lleq$ has such a property.

\begin{proof}[Proof of Lemma \ref{lem-logroup}]
It is well-known that $G = \langle\ints^n, +\rangle$ is a group. To show $G$ is linearly-ordered, it suffices to show that 
$$\forall~ \vec{\vec{x}},\vec{y},\vec{z} \in \ints^n :~~ \vec{x} \lleq \vec{y} \implies \vec{x} + \vec{z} \lleq \vec{y} + \vec{z}~.$$

If $\vec{x} =\vec{y}$ then surely $\vec{x} + \vec{z} = \vec{y} + \vec{z} \implies \vec{x} + \vec{z} \lleq \vec{y} + \vec{z}$.

If instead, $\vec{x} \lleq \vec{y}$, then let $k = \min_{\vec{x}_i < \vec{y}_i} i$. Note that for all $i$ with $1 \leq i < k$, $\vec{x}_i = \vec{y}_i$, and that $\vec{x}_k < \vec{y}_k$. Consider $\vec{x} + \vec{z} $ and $\vec{y}+\vec{z}$. Surely, for all $1 \leq i <k$, $(\vec{x}+\vec{z})_i = (\vec{y}+\vec{z})_i$, since $\vec{x}_i = \vec{y}_i \implies \vec{x}_i + \vec{z}_i = \vec{y}_i + \vec{z}_i$. Likewise, $(\vec{x}+\vec{z})_k < (\vec{y}+\vec{z})_k$, since $\vec{x}_k < \vec{y}_k \implies \vec{x}_k + \vec{z}_k < \vec{y}_k + \vec{z}_k$.

Therefore, $\vec{x} \lleq \vec{y} \implies \vec{x} +\vec{z} \lleq \vec{y} + \vec{z}$.\qed
\end{proof}

\end{document}

%% file: informal.tex
\begin{tikzpicture}[scale=0.45, transform shape]
\tikzstyle{every node}=[minimum size=0.75cm, circle, align=center, font=\LARGE,  draw=black]

\node (agg1) {};
\node [right = 5.5cm  of agg1](agg2) {$v$};
\node [below left = 1cm of agg1] (rack2) {};
\node [below right = 1cm of agg1](rack1) {};
\node [below right = 1 cm of agg2](rack3) {};
\node [below left = 1cm of agg2](rack4) {$u$};
\node [below = of rack1] (srv1) {};
\node [right = 0.5 of srv1] (srv2) {};
\node [below = of rack2] (srv3) {};
\node [right = 0.5 of srv3] (srv4) {};

\node [below = of rack4] (srv6) {};
\node [below = of rack3] (srv8) {};
\node [right = 0.5 of srv8] (srv7) {};
\node [right = 0.5 of srv6] (srv9) {};
\node [left = 0.5 of srv6] (srv0) {};

\node [color=white, text= black, below = 0.25 cmof srv6] (C) {$B$};
\node [color=white, text= black, below = -.25cm of C] (B) {$C$};
\node [color=white, text= black, below = 0.25cm of srv0] (A) {$A$};

\node [color=white, text= black, right = 0cm of srv7, text width=1cm] (cap1) {Servers};
\node [color=white, text= black, above = 0.25cm of cap1, text width=1cm] (cap2) {Racks};
\node [color=white, text= black, above = 0cm of cap2, text width=1cm] (cap3) {Switches};

\node [color=white, text=black, right = -0.25cm of C, yshift=-0.5cm] (cap4) {$\Bigg\}$ Replicas};

\draw[->] (agg1) -- (rack1);
\draw[->] (agg1) -- (rack2);
\draw[->] (agg2) -- (rack3);
\draw[->] (agg2) -- (rack4);
\draw[->] (rack1) -- (srv1);
\draw[->] (rack1) -- (srv2);
\draw[->] (rack2) -- (srv3);
\draw[->] (rack2) -- (srv4);

\draw[->] (rack4) -- (srv6);
\draw[->] (rack3) -- (srv9);
\draw[->] (rack3) -- (srv7);
\draw[->] (rack3) -- (srv8);
\draw[->] (rack4) -- (srv0);

\end{tikzpicture}

%% file: informalb.tex
\begin{tikzpicture}[scale=0.45, transform shape]
\tikzstyle{every node}=[minimum size=0.75cm, circle, align=center, font=\LARGE,  draw=black]

\node (agg1) {};
\node [right = 5.5cm  of agg1](agg2) {$v$};
\node [below left = 1cm of agg1] (rack2) {};
\node [below right = 1cm of agg1](rack1) {};
\node [below right = 1 cm of agg2](rack3) {};
\node [below left = 1cm of agg2](rack4) {$u$};
\node [below = of rack1] (srv1) {};
\node [right = 0.5 of srv1] (srv2) {};
\node [below = of rack2] (srv3) {};
\node [right = 0.5 of srv3] (srv4) {}; 

\node [below = of rack4] (srv6) {};
\node [below = of rack3] (srv8) {};
\node [right = 0.5 of srv8] (srv7) {};
\node [right = 0.5 of srv6] (srv9) {};
\node [left = 0.5 of srv6] (srv0) {};

\node [color=white, text= black, below = 0.25cm of srv2] (B) {$A$};
\node [color=white, text= black, below = 0.25 cmof srv6] (A) {$B$};
\node [color=white, text= black, below = 0.25cm of srv9] (C) {$C$};

\node [color=white, text= black, right = 0cm of srv7, text width=1cm] (cap1) {Servers};
\node [color=white, text= black, above = 0.25cm of cap1, text width=1cm] (cap2) {Racks};
\node [color=white, text= black, above = 0cm of cap2, text width=1cm] (cap3) {Switches};

\node [color=white, text=black, right = 0cm of C, yshift=-0.1cm] (cap4) {Replicas};

\draw[->] (agg1) -- (rack1);
\draw[->] (agg1) -- (rack2);
\draw[->] (agg2) -- (rack3);
\draw[->] (agg2) -- (rack4);
\draw[->] (rack1) -- (srv1);
\draw[->] (rack1) -- (srv2);
\draw[->] (rack2) -- (srv3);
\draw[->] (rack2) -- (srv4);

\draw[->] (rack4) -- (srv6);
\draw[->] (rack3) -- (srv9);
\draw[->] (rack3) -- (srv7);
\draw[->] (rack3) -- (srv8);
\draw[->] (rack4) -- (srv0);

\end{tikzpicture}

%% file: proof-figure-thm-1.tex
\begin{tikzpicture}[scale=0.525, transform shape]
\tikzstyle{every node}=[minimum size=1cm, align=center, font=\LARGE]
\tikzset{>=latex}
\tikzset{snake it/.style={decorate, decoration=snake},
                   postaction={decoration={markings,mark=at position 1 with {\arrow{>}}},decorate}}

\pgfmathsetmacro{\yoffset}{-1.5}
\node (auv)[circle, draw=black] at (0, 0) {$a_{uv}$};
\node (av)[circle, draw=black] at (1.5, \yoffset) {$a_v$};
\node (au)[circle, draw=black] at (-1.5, \yoffset) {$a_u$};
\node (u)[circle, draw=black] at (-3, 3*\yoffset) {$u$};
\node (v)[circle, draw=black] at (3, 3*\yoffset) {$v$};

\draw[-] (auv) -- (av);
\draw[-] (auv) -- (au);
\draw[snake it] (au) -- (u);
\draw[snake it] (av) -- (v);
\draw[snake it] (auv) -- (1,1);
\draw[->] (v) -> (u) node [midway, fill=white] {swap};
\end{tikzpicture}

%% file: proof-figure-thm-2.tex
\begin{tikzpicture}[scale=0.46, transform shape]
\tikzstyle{every node}=[minimum size=1cm, align=center, font=\LARGE]
\tikzset{>=latex}
\tikzset{snake it/.style={decorate, decoration=snake},
                   postaction={decoration={markings,mark=at position 1 with {\arrow{>}}},decorate}}

\pgfmathsetmacro{\yoffset}{-1.5}
\node (auv)[circle, draw=black] at (0, 0) {$u$};
\node (av)[circle, draw=black] at (1.5, 0.75*\yoffset) {$w$};
\node (au)[circle, draw=black] at (-1.5, 0.75*\yoffset) {$v$};
\node (rv)[circle, draw=black] at (2.25, 2*\yoffset) {$r_w$};
\node (rw)[circle, draw=black] at (-2.25, 2*\yoffset) {$r_v$};
\node (u)[circle, draw=black] at (-3, 3.5*\yoffset) {$q_v$};
\node (v)[circle, draw=black] at (3, 3.5*\yoffset) {$q_w$};

\draw[-] (auv) -- (av);
\draw[-] (auv) -- (au);
\draw[snake it] (au) -- (rw);
\draw[snake it] (av) -- (rv);
\draw[snake it] (rw) -- (u);
\draw[snake it] (rv) -- (v);
\draw[snake it] (auv) -- (1,1);
\draw[->] (v) -> (u) node [midway, fill=white] {swap};
\end{tikzpicture}

%% file: degenerate-case.tex
\begin{tikzpicture}[scale=0.525]

\node[inner sep=0]
{\includegraphics[scale=0.5]{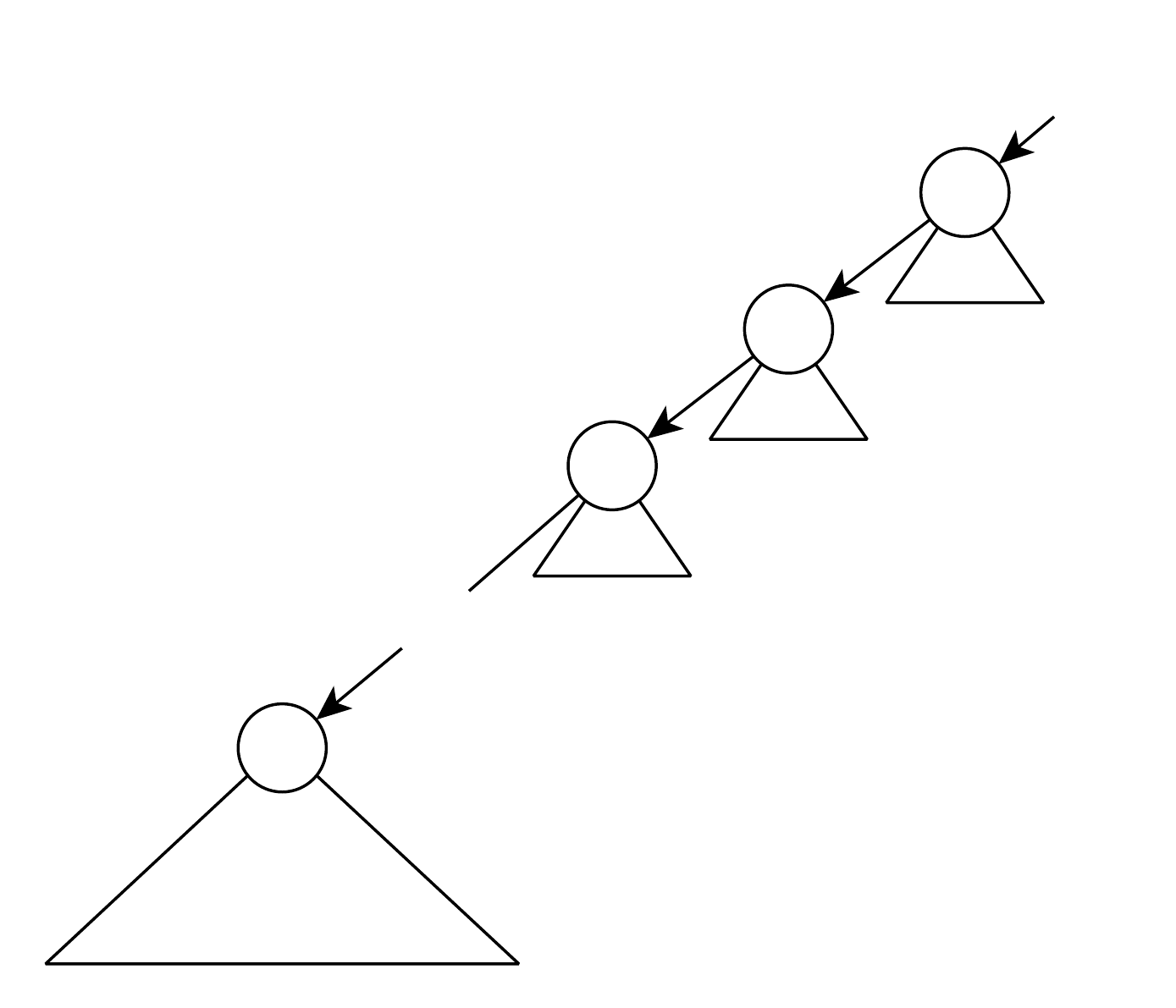}};

\draw (-1.75, -1.1) node[circle, fill=white] {$\iddots$};
\draw[decorate, decoration={brace, amplitude=5pt}] (-4, -2) -- (4.0, 4.3) node [midway, xshift=-8, yshift=8, rotate=37] {length $O(\repfact)$};
\draw (-3.5, -4.4) node {\small $T_t$};
\draw (-3.5, -2.9) node {\small $v_t$};
\draw (0.35, -0.5) node {\small $T_3$};
\draw (0.3, 0.4) node {\small $v_3$};
\draw (2.3, 1.95) node {\small $v_2$};
\draw (2.4, 1.1) node {\small $T_2$};
\draw (4.4, 2.65) node {\small $T_1$};
\draw (4.35, 3.5) node {\small $v_1$};
\draw(4.3, 1.9) node {\tiny $O(1)$ leaves};
\draw(2.3, 0.4) node {\tiny $O(1)$ leaves};
\draw(0.2, -1.1) node {\tiny $O(1)$ leaves};

\end{tikzpicture}

%% file: contracted-nodes.tex
\begin{tikzpicture}[scale=0.525]
\centering
\node[inner sep=0]
{\includegraphics[scale=0.5]{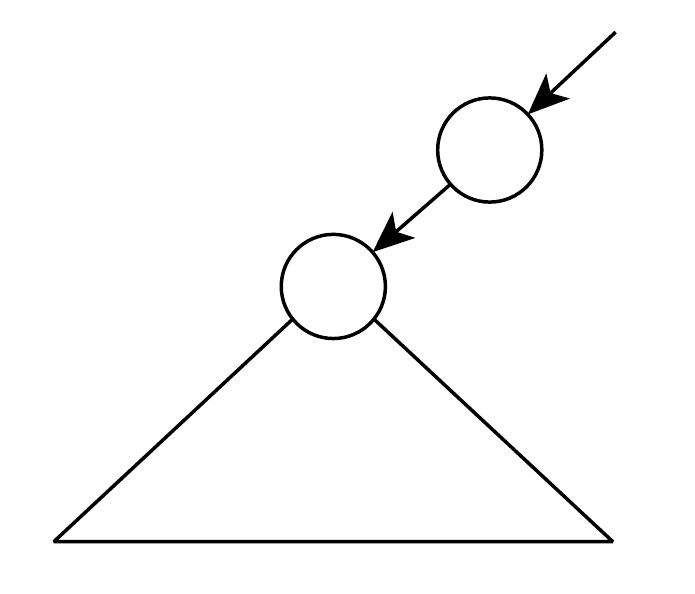}};

\draw (-0.125, -1.5) node[align=center] {$T_t$};
\draw (-0.125, 0) node[align=center] {$v_t$};

\draw (1.35, 1.35) node[align=center] {$w$};


\end{tikzpicture}

%% file: wads-2015.bbl
\begin{thebibliography}{10}
\providecommand{\url}[1]{\texttt{#1}}
\providecommand{\urlprefix}{URL }

\bibitem{BakWyl+:2002:TR}
Bakkaloglu, M., Wylie, J.J., Wang, C., et. al: {On correlated failures in
  survivable storage systems}. Tech. Rep. CMU-CS- 02-129, Carnegie Mellon
  University (2002)

\bibitem{BluEas+:2011:FOCS}
Blume, L., Easley, D., Kleinberg, J., Kleinberg, R., Tardos, E.: Which networks
  are least susceptible to cascading failures? In: Foundations of Computer
  Science (FOCS), 2011 IEEE 52nd Annual Symposium on. pp. 393--402 (Oct 2011)

\bibitem{lian-chen-zhang}
Chen, M., Chen, W., Liu, L., et. al: An analytical framework and its
  applications for studying brick storage reliability. In: Proceedings of 26th
  International Symposium on Reliable Distributed Systems. pp. 242--252. IEEE
  Computer Society (2007)

\bibitem{ForFra+:2010:OSDI}
Ford, D., Labelle, F., Popovici, F.I., Stokely, M., Truong, V.A., Barroso, L.,
  Grimes, C., Quinlan, S.: Availability in globally distributed storage
  systems. In: Presented as part of the 9th USENIX Symposium on Operating
  Systems Design and Implementation. USENIX, Berkeley, CA (2010)

\bibitem{HuJia+:2001:JPDC}
Hu, X.D., Jia, X.H., Du, D.Z., Li, D.Y., Huang, H.J.: Placement of data
  replicas for optimal data availability in ring networks. Journal of Parallel
  and Distributed Computing  61(10),  pp. 1412 -- 1424 (2001)

\bibitem{KimDob:2010:TransRel}
Kim, J., Dobson, I.: {Approximating a loading-dependent cascading failure model
  with a branching process}. {IEEE} Transactions on Reliability  59,  pp.
  691--699 (Dec 2010)

\bibitem{DBLP:conficdcsLianCZ05}
Lian, Q., Chen, W., Zhang, Z.: {On the impact of replica placement to the
  reliability of distributed brick storage systems}. In: 25th IEEE
  International Conference on Distributed Computing Systems (ICDCS'05). pp. pp.
  187--196. IEEE (2005)

\bibitem{NatYu+:2006:NSDI}
Nath, S., Yu, H., Gibbons, P.B., Seshan, S.: Subtleties in tolerating
  correlated failures in wide-area storage systems. In: 3rd Symposium on
  Networked Systems Design and Implementation {(NSDI} 2006), May 8-10, 2007,
  San Jose, California, USA, Proceedings. pp. 225--238 (2006)

\bibitem{NieLui+:2014:IPL}
Nie, L., Liu, J., Zhang, H., Xu, Z.: On the inapproximability of minimizing
  cascading failures under the deterministic threshold model. Information
  Processing Letters  114,  pp. 1--4 (2014)

\bibitem{Pezoa}
Pezoa, J., Hayat, M.: Reliability of heterogeneous distributed computing
  systems in the presence of correlated failures. IEEE Transactions on Parallel
  and Distributed Systems,  25(4),  pp. 1034--1043 (April 2014)

\bibitem{Ple:2013:Blog}
Pletz, J.: {The price of failure: Data-center power outage cost sears \$2.2m in
  profit}. http://www.chicagobusiness.com/article/20130604/BLOGS11/130609948/
  the-price-of-failure-data-center-power-outage-cost-sears-2-2m-in-profit (Jun
  2013)

\bibitem{SheWu:2001:TCS}
Shekhar, S., Wu, W.: Optimal placement of data replicas in distributed database
  with majority voting protocol. Theoretical Computer Science  258,  pp. 555 --
  571 (2001)

\bibitem{Ver:2013:Blog}
Verge, J.: {How a switch failure in utah took out four big hosting providers}.
  http://www.datacenterknowledge.com/archives/2013/08/05/how-did-the-failure-of-network-switches-at-a-little-known-data-center-in-provo-utah-knock-four-major-services-and-millions-of-web-pages-offline
  (Sep 2013)

\bibitem{WeaMos+:2002:SRDS}
Weatherspoon, H., Moscovitz, T., Kubiatowicz, J.: Introspective failure
  analysis: Avoiding correlated failures in peer-to-peer systems. In:
  Proceedings of the 21st IEEE Symposium on Reliable Distributed Systems, 2002.
  pp. 362--367 (2002)

\bibitem{ZhaWu+:2009:JPDC}
Zhang, Z., Wu, W., Shekhar, S.: Optimal placements of replicas in a ring
  network with majority voting protocol. J. Parallel Distrib. Comput.  69(5),
  461--469 (May 2009)

\bibitem{ZhuYan+:2014:TPDS}
Zhu, Y., Yan, J., Sun, Y., He, H.: Revealing cascading failure vulnerability in
  power grids using risk-graph. IEEE Transactions on Parallel and Distributed
  Systems,  25(12),  pp. 3274--3284 (Dec 2014)

\end{thebibliography}
